\documentclass[a4paper, 10pt, conference]{ieeeconf}
\IEEEoverridecommandlockouts
\usepackage{amsmath,rotating}
\usepackage{cases}
\usepackage{hyperref}
\usepackage{amsfonts}
\usepackage{amssymb}
\usepackage{times,subfigure}
\usepackage{comment}
\usepackage{algpseudocode}
\usepackage{mathtools}
\usepackage{algorithm}
\usepackage{graphicx}
\usepackage{epstopdf}
\usepackage[T1]{fontenc}
\usepackage[scanall]{psfrag}

\newtheorem{Proposition}{Proposition}

\newtheorem{Remark}{Remark}
\usepackage{epstopdf}
\usepackage{graphicx}

\def\sqw{\hfill\hbox{\lower.1ex\hbox{$\sqcup$}
    \kern-1.02em\lower.1ex\hbox{$\sqcap$}}\ }

\newtheorem{Definition}{Definition}

\newcommand{\yv}{\mathbf{y}}

\newcommand{\bv}{\mathbf{b}}
\newcommand{\cv}{\mathbf{c}}

\newcommand{\xv}{\mathbf{x}}

\newcommand{\ev}{\mathbf{e}}

\newcommand{\rv}{\mathbf{r}}
\newcommand{\sv}{\mathbf{s}}

\newcommand{\R}{\mathit{R}}

\title{\LARGE \bf
Fully distributed PageRank computation with exponential convergence
}
\author{Liang Dai$^{1}$ and Nikolaos M. Freris%
\thanks{$^{1}$ Major part of the work was done when the first author was with the Engineering Division of NYU Abu Dhabi. Correspondence can be directed to E-mail: 
        {\tt\small liang.dai.just@gmail.com}}%
        }

\begin{document}
\maketitle
\begin{abstract}
    This work studies a fully distributed algorithm for computing the PageRank vector, which is inspired by the Matching Pursuit and features: 1) a fully distributed implementation 2) convergence in expectation with exponential rate 3) low storage requirement (two scalar values per page). Illustrative experiments are conducted to verify the findings.
\end{abstract}

\IEEEpeerreviewmaketitle

\begin{section}{Problem Statement}
PageRank vector was proposed by the founders of the Google to quantify the importance rankings of the webpages of the Internet \cite{c1,c2}. Due to the generality of the idea, PageRank has been extended to application in Biology, Chemistry and some other domains, more details can be found in the review paper  \cite{beyondpagerank}.

  Suppose there are $N$ pages in a network. The connectivity (i.e., the topology induced by the hyperlinks present in websites) can be characterized by the
hyperlink matrix $A\in \R^{N\times N}$ defined as follows: its $(i,j)$-th element is $\frac{1}{N_j}$, if there is a link from page-$j$ to page-$i$, where $N_j$ denotes the number of outgoing links of page $j$ (the number of pages which the page $j$ points to); otherwise $A_{i,j} = 0$. By construction $A$ is a non-negative, column stochastic matrix (i.e., a matrix with non-negative elements and each column summing up to one). In this work, we assume without any loss of generality that there are no dangling pages (i.e. pages with no outgoing pages), i.e., $A$ has no zero columns.

One potential choice for the PageRank vector is the normalized principal eigenvector (with all the elements summing up to one) of matrix $A$. However, one drawback of such choice is that, when the network is not fully connected, the principal eigenvector of matrix $A$ may not be unique.

To overcome this problem, a 'perturbed' version of $A$ is adopted for defining the PageRank vector, given by 
$$M = \alpha A + (1-\alpha)S,$$
where $S = \frac{1}{N}\mathbf{1}\mathbf{1}^{T}$, with $\mathbf{1} = [1,1,\cdots,1]^{T} \in \R^{N\times 1} $ and $\alpha \in (0,1)$. The suggested value for $\alpha$ is $0.85$ \cite{c1}. The original PageRank vector is defined as:
\vspace{1mm}
\begin{Definition}[PageRank]
\label{def1}
Given the perturbed hyperlink matrix  $M$, the vector $\xv^{*}$ is the unique vector $\xv^{*}$ that satisfies:
\begin{enumerate}
\item
$M\xv^{*} = \xv^{*},$
\item
$\sum_{i=1}^{N} \xv^{*}_i =1$ and $\xv^{*} \ge 0$.
\end{enumerate}
\end{Definition}
\vspace{1mm}

Note that since $M$ is a positive, column stochastic, and irreducible matrix, the Perron-Frobenius Theorem \cite{c6} guarantees existence and uniqueness of a positive right-eigenvector with corresponding eigenvalue equal to 1, which is precisely $\xv^{*} $. The second property is simply a normalization of its entries to sum up to one.  

It is plain to see that the ranking will not be affected by a positive rescaling of the PageRank vector. In our work, we will adopt the following (positively rescaled by the network size $N$) scaled PageRank vector. The main advantage of this choice is that computations will not involve the network size $N$ and will be made clear in the sequel. 
\vspace{1mm}
\begin{Definition}[Scaled-PageRank]
\label{def2}
Given the perturbed hyperlink matrix  $M$, the scaled PageRank vector $\xv^{*}$ is the vector which satisfies:
\begin{enumerate}
\item
$M\xv^{*} = \xv^{*},$
\item
$\sum_{i=1}^{N} \xv^{*}_i =N$ and $\xv^{*} \ge 0$.
\end{enumerate}
\end{Definition}
\vspace{1mm}

As the internet is of huge scale, it becomes very difficult to save the entire matrix $M$ and solve  $M\xv = \xv$ in a single machine. This is performed by Google on a regular basis using the centralized power iteration \cite{c2} which requires large storage and computational power.  Additionally, a change in matrix $M$ (for example the creation or deletion of a website or changes in the hyperlinks present in a page) typically entails re-computation of the PageRank vector from scratch.

To overcome the difficulties, several distributed methods ( where each page updates its PageRank value by exchanging the information only with neighbouring pages, i.e., pages that it links to or pages that link to it) have been suggested for this problem. Based on the idea of Monte Carlo simulations, \cite{c0} proposed the following approach: starting from each node, the algorithm performs multiple rounds of random walks via certain absorbing Markov chains, and PageRank vector is estimated by the frequency of visits to this node from all the random walks. The method features fast convergence as well as distributed implementation, however, the simultaneous runs of a large number of random walks may lead to the problem of congestion in the network. In the following, we will focus on reviewing the ideas based on linear algebraic techniques. In \cite{c3}, a randomized distributed algorithm was proposed based on stochastic  power iterations together with the Polyak averaging scheme.   Recently, based on an application of the Stochastic Approximation (SA) framework \cite{sa}, a randomized distributed algorithm was designed \cite{csa}. Nonetheless, in both \cite{c3} and \cite{csa}, during each update, a webpage needs to request information from its in-coming neighbours (i.e. the set of webpages that link to it), which might impose practical limitations in that: 1) it either requires additional storage of a list of incoming neighbours, which sometimes could be of huge size; 2) or it might incur delays (for example, wait till all the information has been transimitted) in obtaining the values from the incoming neighbours. Furthermore, the approaches in \cite{c3} and \cite{csa} are of (or can be reformulated as) SA-type algorithms, which feature sub-exponential convergence rate \cite{sapagerankp1, csa}.  In \cite{increpage}, a randomized incremental optimization based distributed algorithm was proposed: nonetheless, similarly to the work in \cite{c3} and \cite{csa}, information from in-coming pages are required for the algorithm's updates.
 
In this work, seeking to overcome these issues, we propose a fully distributed algorithm (in which updating
webpages only use the PageRank values of outgoing pages, while also no knowledge of the network size is required) with provable exponential convergence (in expectation). From a signal decomposition point of view the proposed method can be seen as randomized Matching Pursuit algorithm. The main attributes of the new algorithm are: 
\begin{enumerate}
\item
 It uses only the knowledge of the out-going webpages and no knowledge of the network size is assumed;
\item
It converges exponentially fast, in expectation;
\item
It only requires storing two scalar values per webpage (the PageRank estimate along with a residual value, explicated below).
\end{enumerate}
\end{section}

\begin{section}{The proposed algorithm}
In the following, $U[m,n]$ will be used to denote the uniform sampling of a natural number between $m$ and $n$. Conventions in {\em Matlab} will be used to denote the rows and columns of a matrix. $I$, $\mathbf{1}$ and $\mathbf{0}$ denote the identity matrix,  all-one vector and all-zero vector respectively, where dimension will be made clear from the context. $\ev_{k}$ denotes the $k$-th unit vector (1 in $k$-th entry and 0 elsewhere), while $\| \cdot \|$ represents the $l_2$ norm of a vector. 
 \subsection{Problem Reformulation}

Substituting $M$ to the definition of PageRank vector in Definition $\ref{def2}$, and using the property of matrix $S$ that $S\xv = 1$ for any $\xv$ with $\sum_{i}x_i =1$, we have the following equivalent characterization of the scaled PageRank vector
\begin{subnumcases}{}
(I - \alpha A) \xv^{*} = (1 - \alpha)\mathbf{1}, \label{eq.case1}  \\
\sum_{i=1}^{N} \xv^{*}_i =N \,\, \text{and}\,\,  \xv^{*} \ge 0. \label{eq.case2}
\end{subnumcases}

From $\eqref{eq.case1}$, we get the vector 
\begin{align}
\label{eq.page1}
\hat{\xv} = (1 - \alpha)(I - \alpha A)^{-1}\mathbf{1}.
\end{align}

If we can further establish that all the elements of vector $\hat{\xv}$ are nonnegative and summing up to one, then $\hat{\xv}$ will be the PageRank vector as in the Definition $\ref{def2}$. The following proposition confirms this.

\begin{Proposition}
\label{pro1}
The scaled PageRank vector is given as
\begin{align}
\label{eq.page2}
\xv^{*} = (1 - \alpha )(I - \alpha A)^{-1}\mathbf{1}.
\end{align}
\end{Proposition} 
 \begin{proof}
From the Gershgorin circle theorem \cite{c6}, the eigenvalues of $A$ all have magnitudes in $[0,1]$, which implies that $I - \alpha A$  is invertible (since $0< \alpha <1$) and its inverse is given by
\begin{align}
\label{eq.series}
(I - \alpha A)^{-1} = \sum_{k=0}^{\infty} \alpha^{k}A^{k}.
\end{align}  

Right-multiplying by $\mathbf{1}$ and using the non-negativity of $A$, we get that $\xv^{*} > 0$. We are left to verify that 
\begin{align}
\mathbf{1}^{T}(1 - \alpha )(I - \alpha A)^{-1}\mathbf{1} = N.
\end{align}

For any $k\ge 0$, since $A$ is a column stochastic matrix, we have that
$\mathbf{1}^{T}A^{k}\mathbf{1} =N,$ thereby 
\begin{align*}
\mathbf{1}^{T}(1 - \alpha )(I - \alpha A)^{-1}\mathbf{1} 
=\,& (1 - \alpha )\left(\sum_{k=0}^{\infty} \alpha^{k}\mathbf{1}^{T}A^{k}\mathbf{1}\right)\\
=\,& N(1 - \alpha )\sum_{k=0}^{\infty} \alpha^{k} =\,N,
\end{align*}
which concludes the proof.
 \end{proof}

 \subsection{The algorithmic description}
 Following Proposition $\ref{pro1}$, we can find the scaled PageRank vector by solving the system of linear equations 
 \begin{equation}
 \label{eq.solve}
 (I - \alpha A)\xv^{*} = (1 - \alpha )\mathbf{1}.
 \end{equation}
 
 Note that there is no more dependency on network size $N$, which means that if we design a distributed solver for $\eqref{eq.solve}$ it will be fully distributed. We will take a signal processing point of view to solve this equation: we regard the columns of matrix $I - \alpha A$ as the atoms (or the basis) of a dictionary, and we are left to find the representation (or decomposition) of vector $(1 - \alpha )\mathbf{1}$ using these atoms. In our case, this representation will be unique. To simplify notations, we let  $$B \triangleq I - \alpha A,$$
 and 
 $$\yv \triangleq (1-\alpha) \mathbf{1} $$

The Matching Pursuit (MP) algorithm in \cite{matchingpursuit} has become a standard tool to find signal representations under a given set of atoms ( termed as dictionary, which is often over-complete).  In brief, the MP algorithm is an iterative algorithm: at each iteration, it identifies the "best matching" atom with the signal residual, and subsequently updates the residual by subtracting its projection to this "best matching" atom.  In our case, even though the dictionary, i.e. the matrix $B$, is not over-complete, we can still apply the MP algorithm to find the unique representation of  vector $(1 - \alpha )\mathbf{1}$ using the columns (the atoms) of $B$.
 
 However, the 'best matching' step in MP algorithm is not amendable to a distributed implementation, as it requires searching all columns, i.e. all webpages in our case of interest. To address this issue, we modify the original MP algorithm via randomization: instead of picking the 'best matching' atom at each iteration, the proposed algorithm picks a random atom from the dictionary and perform the projection step. By doing so, the derived algorithm can be implemented in a fully distributed fashion. In addition, we shall show that the algorithm will converge exponentially fast (in expectation, since we introduce randomization) to the desired solution. The detailed algorithmic description is given in Algorithm $\ref{pagerank.alg}$.   
\begin{algorithm}
\caption{Matching Pursuit based PageRank Computation}
\label{pagerank.alg}
\begin{algorithmic}
\State{\textit{Initialization}: \\Initialize vectors $\rv_0, \xv_0\in \R^{N}$ as $\yv$ and $\mathbf{0}$ respectively.}
\State{\textit{Iterations:}}
\For{$t = 0,1, \cdots,T-1$} 
    \State{Generate $k = U[1,N]$, and update:}
    \begin{align}
    \xv_{t+1} = \xv_{t} + \frac{B(:,k)^{T}\rv_{t}}{\|B(:,k)\|^2} \mathbf{e}_k
     \label{alg2.v2}
    \end{align}
    
    \begin{align}
    \rv_{t+1} = \rv_{t} -\frac{B(:,k)^{T}\rv_{t}}{\|B(:,k)\|^2} B(:,k)
    \label{alg2.v1}
    \end{align}
 \EndFor
\State{\textit{Return:}}
\State{Return $\xv_T$.}
\end{algorithmic}
\end{algorithm}
\begin{Remark}
A fully asynchronous scheme (i.e. the 'exponential clocks' approach) to implement the uniform (or more general) sampling in Algorithm $\ref{pagerank.alg}$ can be found in the 'Implementation Issues' section of  \cite{expclock}, and the references therein. 
\end{Remark}

\begin{Remark}
The sequences of $\{\xv_t\}_{t=0}^{\infty}$ and $\{\rv_t\}_{t=0}^{\infty}$ in Algorithm $\ref{pagerank.alg}$   keep track of the approximations to $\xv^{*}$ and the signal residuals, respectively. Note that the increment $\xv_{t+1}-\xv_{t}$ in equation $\eqref{alg2.v2}$ can be obtained by solving the following coordinate descent optimization problem:
\begin{align*}
& \min_{\Delta} \|B\xv_{t+1} - \yv\|^2 \\ &\text{s.t. }\, \,\xv_{t+1} = \xv_{t} + \Delta\ev_k.
\end{align*}
\end{Remark}

\subsection{Convergence analysis}
Before proceeding further, we define some auxiliary quantities. For each $1\le k\le N$, let $$\bv_k = \frac{B(:,k)}{\|B(:,k)\|},$$  and  $$\hat{B} = [\bv_1,\cdots, \bv_N], \,\, P_k = \bv_k\bv_k^{T},$$ 
it follows that $\hat{B}\hat{B}^{T} = \sum_{k=1}^{N}P_k.$

Notice that the equation $\eqref{alg2.v1}$ can be rewritten as follows
\begin{align*}
 \rv_{t+1} = (I - P_k)\rv_{t},
\end{align*}
which gives that
\begin{align*}
 \|\rv_{t+1}\|^2 = \rv_{t+1}^{T}\rv_{t+1}  = \rv_{t}^{T} (I - P_k)\rv_{t}.
\end{align*}

Conditioned on $\rv_{t}$, we have that
\begin{align*}
 \mathbb{E}[\|\rv_{t+1}\|^2|\rv_{t}] =&\frac{1}{N}\sum_{k=1}^{N} \rv_{t}^{T} (I - P_k)\rv_{t}\\
 =&\rv_{t}^{T} \left(I - \frac{1}{N}\hat{B}\hat{B}^{T}\right)\rv_{t}
\end{align*}

Since $B$ is a full rank square matrix, so $\hat{B}$ as well. This implies that $\sigma(\hat{B})$, the smallest singular value of $\hat{B}$, is nonzero. Using
$$ \frac{1}{N}\hat{B}\hat{B}^{T} \succeq  \frac{\sigma^2(\hat{B})}{N} I ,$$
 we have
\begin{align*}
 \mathbb{E}[\|\rv_{t+1}\|^2|\rv_{t}] 
 \le&\left(1- \frac{\sigma^2(\hat{B})}{N}\right)\|\rv_{t}\|^2.
\end{align*}

Iterating this equation, we get
\begin{align}
\label{eq.expo}
 \mathbb{E}\|\rv_{t}\|^2 \le \left(1- \frac{\sigma^2(\hat{B})}{N}\right)^{t}\|\rv_{0}\|^2,
\end{align}
for any $t\ge 0$, which establishes the exponential decay of the squared norm of the signal residual.

A useful property of Algorithm $\ref{pagerank.alg}$, which will be useful for establishing Proposition $\ref{proposition.conv}$, is that during the run of the algorithm the vector $B\xv_t + \rv_t$ is always kept constant. To see this, multiplying both sides of equation $\eqref{alg2.v2}$ by matrix $B$ gives that
\begin{align*}
    B\xv_{t+1} = B\xv_{t} + \frac{B(:,k)^{T}\rv_{t}}{\|B(:,k)\|^2} B\ev_k,
    \end{align*}
    which simplifies to
\begin{align}
\label{eq.up}
    B\xv_{t+1} = B\xv_{t} + \frac{B(:,k)^{T}\rv_{t}}{\|B(:,k)\|^2} B(:,k).
    \end{align}
    
Adding equation $\eqref{eq.up}$ and equation $\eqref{alg2.v1}$, we have that
\begin{align*}
    B\xv_{t+1} + \rv_{t+1}= B\xv_{t} + \rv_{t},
\end{align*}
which gives that 
\begin{align}
\label{eq.conservation}
B\xv_{t} + \rv_{t} = \rv_{0} = \yv
\end{align}
for any $t\ge 0.$

Summarizing the conservation property in equation $\eqref{eq.conservation}$ and the exponential decreasing fact in equation $\eqref{eq.expo}$ gives:
\begin{Proposition}\label{proposition.conv}
The vector sequence $\{\xv_t\}_{t=0}^{\infty}$  converges to the scaled PageRank vector $\xv^{*}$ (in Definition 2) with the expected exponential rate. In specific, for all $t\ge 0$
\begin{align}
\label{eq.expo_2}
 \mathbb{E}\|\xv_{t} - \xv^{*}\|^2 \le \sigma^{-2}(\hat{B})\|\rv_{0}\|^2\left(1- \frac{\sigma^2(\hat{B})}{N}\right)^{t}.
\end{align}
\end{Proposition} 
\begin{proof}
According to Algorithm $\ref{pagerank.alg}$, we have that $\rv_0 = B\xv^{*}$. Substituting it into equation $\eqref{eq.conservation}$ and rearranging the terms, we obtain
\begin{align*}
B(\xv_{t}-\xv^{*}) = \rv_{t}.
\end{align*}

From $\eqref{eq.expo}$, it follows that
\begin{align*}
\|B(\xv_{t}-\xv^{*}) \|^2 \le \|\rv_{0}\|^2\left(1- \frac{\sigma^2(\hat{B})}{N}\right)^{t}.
\end{align*}

Additionally, the definition of the smallest singular value gives that
\begin{align*}
\|B(\xv_{t}-\xv^{*}) \|^2 \ge \sigma^2(\hat{B}) \|\xv_{t}-\xv^{*}\|^2,
\end{align*}
which concludes the proof.
\end{proof}

\subsection{Distributed implementation}
In this section, for a vector $\cv_i$, we use $\cv_{i,j}$ to denote its $j$-th element. For web page $k$, the indices of its outgoing pages are denoted as $$\mathcal{N}_k = \{n_1,\cdots,n_{N_k}\},$$where $N_k$ is the number of outgoing pages. 

We first show the distributed implementation of equation $\eqref{alg2.v2}$.  Since only the $k$-th element of $\ev_k$ is nonzero, to update $\xv_{t+1}$, we only need to update $\xv_{t+1,k}$ (the value of the currently triggered page $k$) according to 
    \begin{align}
    \xv_{t+1,k} = \xv_{t,k} + \frac{B(:,k)^{T}\rv_{t}}{\|B(:,k)\|^2},
     \label{alg2.v3}
    \end{align}
    and $\xv_{t+1,j} = \xv_{t,j}$ for all $j\ne k$ and $1\le j \le N$ (all other pages keep their previous estimates).
    
    The numerator in $\eqref{alg2.v3}$ is 
      \begin{align*}
    B(:,k)^{T}\rv_{t}  = &\left(\ev_k - \alpha A(:,k)\right)^{T}\rv_{t}\\
  = &\rv_{t,k} -\alpha\frac{\sum_{j=1}^{N_k}\rv_{t,n_j}}{N_k},
      \end{align*}  
        which can be computed by \emph{reading} the residual values of $\{\rv_{t,n_j}\}_{j=1}^{N_k}$ from all the outgoing neighbours of page $k$.
        
        The denominator in $\eqref{alg2.v3}$ is given as 
              \begin{align*}
    \|B(:,k)\|^2  = &\left(\ev_k - \alpha A(:,k)\right)^{T}\left(\ev_k - \alpha A(:,k)\right)\\
  = & 1-2\alpha A_{k,k} + \alpha^2\|A(:,k)\|^2, \\
  =& 1-2\alpha A_{k,k} + \frac{\alpha^2}{N_k},
        \end{align*}  
        which can be computed by knowing the local information $N_k$ and $A_{k,k}$. Note that $A_{k,k}=0$ if the page $k$ does not link to itself and $A_{k,k} = \frac{1}{N_k}$ otherwise. 
        \begin{Remark}
        $ \{\|B(:,k)\|^2\}_{k=1}^N$ can be calculated in a pre-processing step to avoid recalculation at every iteration. 
        \end{Remark}

        Next, we show distributed computation of equation $\eqref{alg2.v1}$. Note that, we have just shown distributed computation for the values $B(:,k)^{T}\rv_{t}$ and $\|B(:,k)\|^2$. Therefore, we have: for each $n_j\in \mathcal{N}_k$ and $n_j\ne k$
   \begin{align*}
    \rv_{t+1,n_j} = &\rv_{t,n_j} - \frac{B(:,k)^{T}\rv_{t}}{\|B(:,k)\|^2}(-\frac{\alpha}{N_k})\\
= &\rv_{t,n_j} + \frac{\alpha}{N_k} \frac{N_k\rv_{t,k} -\alpha\sum_{j=1}^{N_k}\rv_{t,n_j}}{N_k + \alpha^2-2\alpha N_k A_{k,k} },
    \end{align*}
    which can also be computed solely using information from page $k$'s outgoing links and its local information.

    If $k\in \mathcal{N}_k$, i.e., page $k$ has a link to itself, the update of $\rv_{t+1,k}$ is given as follows
       \begin{align*}
    \rv_{t+1,k} = &\rv_{t,k} - \frac{B(:,k)^{T}\rv_{t}}{\|B(:,k)\|^2}(1-\frac{\alpha}{N_k})\\
= &\rv_{t,k} -(1- \frac{\alpha}{N_k}) \frac{N_k\rv_{t,k} -\alpha\sum_{j=1}^{N_k}\rv_{t,n_j}}{N_k + \alpha^2-2\alpha},
    \end{align*}
    otherwise 
           \begin{align*}
    \rv_{t+1,k} = &\rv_{t,k} \text{  for  } j\notin \mathcal{N}_k,
  \end{align*}
    which are all distributed implementable as well.

   To summary, our distributed implementation picks a single page at each iteration, reads the residuals from its outgoing pages, and updates selected page's PageRank estimate, as well as the residuals of the outgoing pages. Therefore, at each iteration, the number of 'reads' and 'writes' is exactly equal to the number of outgoing webpages of the selected webpage. 
\end{section}

\begin{section}{Illustrative Experiments}
In this part, we will conduct one synthesised example to verify the findings in the previous sections. The hyperlink matrix $A$ is generated as follows: We first generate a $N\times N (N = 100)$ random matrix, with entries  i.i.d. generated following a uniform distribution in $[0,1]$. Each element is then thresholded with a given constant, which is set to be 0.5 in this experiment. The $\alpha$ is chosen as 0.85. To illustrate the exponential decreasing result in Proposition $\ref{proposition.conv}$, we run 100 rounds of simulations and then average them. The method is compared with the method in \cite{c3} (initialized with an all one vector) and the method in \cite{increpage} (initialized with a zero vector).  The results are reported in Figure 1, and the analysis is given in its caption.

\begin{figure}[ht!] 
\label{fig.experiment}
\centering
\centering
\includegraphics[width=0.4\textwidth]{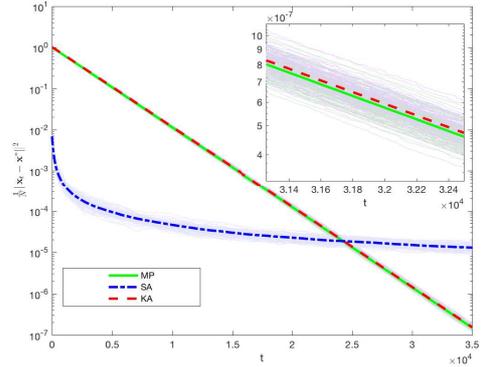}
    \caption
    {This figure illustrates the trajectories of $\frac{1}{N}\|\xv_t - \xv^{*}\|^2$ as well as their averaged trajectories. The solid green and dot red lines represent the averaged trajectory for the proposed Matching-Pursuit based method and the method in \cite{increpage}, which decreases exponentially with a similar rate; the dash-dot blue line shows the averaged trajectory for the work in \cite{c3}, which decreases sub-exponentially. Also note that the variance of the trajectories of method in \cite{c3} is large than the trajectories obtained by the other two approaches.}
\end{figure}
\end{section}

\begin{section}{Discussions}
In this note, we proposed a fully distributed scheme to compute the PageRank vector. The method activates one random webpage at each step, then communicates and updates information with its outgoing webpages. It also converges fast with an exponential speed in expectation. There are several questions for future work: 1) parallelization of the algorithm; 2) generalization to  a dynamic network setting; 3) improvement by a non-uniform sampling; 4) stopping criteria: when the iterations can be terminated to certify a correct ranking.
\end{section}

\begin{section}{Appendix -- Network size estimation}

We make the additional assumption that network is fully connected.  Note that $\mathbf{s} = \frac{1}{N}\mathbf{1} \in \R^{N}$ is the eigenvector of matrix $A^{T}$ corresponding to its principal eigenvalue (which is 1), hence $A^{T}\mathbf{s} = \mathbf{s}$, i.e. $$ (I - A)^{T}\mathbf{s} =0.$$  Let $C \triangleq (I - A)^{T}$, it is evident that $C\mathbf{s} = 0$ and its nullspace is of dimension 1 (under the assumption of network strong connectivity). This implies that,  to find $\mathbf{s}$, we can find a vector which lies in the nullspace of $C$ with its entries summing up to 1.  Algorithm 2 will return back such vector.  The intuition behind Algorithm 2 is that, it starts with a vector with its entries summing up to 1, and then iteratively subtracts out its projections on the row vectors of $C$, which eventually will give a vector in the nullspace of $C$. Note that during the run of iterations, the summation of the entries of $\{\sv_t\}$ will remain unchanged, which can be easily verified by multiplying both sides of equation $\eqref{algsize.v2}$ by $\mathbf{1}^{T}$. Once the estimated vector $\hat{\sv}$ is of $\sv$ is obtained, each webpage, say page $i$, can estimate the web size as $\frac{1}{\hat{\sv}_i}$. 

It is important to note that the computation in $\eqref{algsize.v2}$ also only requires communications with the webpage's out-going links, hence $\eqref{algsize.v2}$ can be distributed implemented in the same way as in Algorithm 1. An exponential convergence in mean can be established for the sequence $\|\sv_t - \sv\|^2$, and the key steps are given as follows. From equation $\eqref{algsize.v2}$, it follows that
\begin{align*}
 \sv_{t+1} -\sv = (I - C_k)(\sv_{t}-\sv),
\end{align*}
where $C_k \triangleq \frac{C(k,:)^{T}C(k,:)}{\|C(k,:)\|^2}$, which implies
\begin{align*}
 \mathbb{E}[\|\sv_{t+1} -\sv \|^2|\sv_t] = &(\sv_{t} -\sv)^{T}\left(I- \frac{1}{N}\sum_{k=1}^{N}C_k\right)(\sv_{t} -\sv). 
\\ \le &  \left(I - \frac{\sigma_2(\hat{C})}{N} \right)\|\sv_{t} -\sv \|^2, 
\end{align*}
where $\sigma_2(\hat{C})$ denotes the second smallest singular value of matrix $\hat{C} \triangleq \sum_{k=1}^{N}C_k.$  Note that  the smallest singular value of $\hat{C}$ will be zero. The inequality follows directly from decomposing $\sv_{t} -\sv$ into two parts: its projection on $\sv$ and the residual. 
\begin{algorithm}
\caption{Network Size Estimation}
\label{websize.alg}
\begin{algorithmic}
\State{\textit{Initialization}: \\Initialize vectors $\sv_0\in \R^{N}$ as $[1,0,\cdots,0]$.}
\State{\textit{Iterations:}}
\For{$t = 0,1, \cdots,T-1$} 
    \State{Generate $k = U[1,N]$, and update:}
    \begin{align}
    \sv_{t+1} = \sv_{t} - \frac{C(k,:)^{T}\sv_{t}}{\|C(k,:)\|^2} C(k,:)
     \label{algsize.v2}
    \end{align}
 \EndFor
\State{\textit{Return:}}
\State{Return $\sv_T$.}
\end{algorithmic}
\end{algorithm}
 
 \begin{figure}[ht!] 
\label{fig.size}
\centering
\centering
\includegraphics[width=0.4\textwidth]{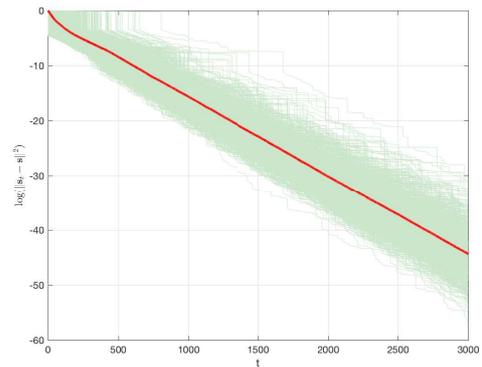}
    \caption
    {This figure illustrates the behaviour of $\|\sv_t - \sv\|^2$. The network is generated in the same way as in section III. We run the experiment 1000 times, and the thick red line illustrate the average trajectory which decreases with an exponential speed. }
\end{figure}
\end{section}

\end{document}